\documentclass[11pt]{article}
\usepackage{geometry}
\usepackage[parfill]{parskip}

\usepackage{graphicx}

\usepackage{amssymb}

\usepackage{epstopdf}

\usepackage{amsmath,amsthm,amscd,amssymb}
\usepackage{latexsym}
\numberwithin{equation}{section}

\theoremstyle{plain}
\newtheorem{theorem}{Theorem}[section]
\newtheorem{lemma}[theorem]{Lemma}

\newtheorem{proposition}[theorem]{Proposition}
\newtheorem{conjecture}[theorem]{Conjecture}

\theoremstyle{definition}

\theoremstyle{remark}

\newtheorem{case[theorem]}{Case}

\title{On the CNF-complexity of bipartite graphs containing no $K_{2,2}$'s}
\author{Nets Hawk Katz}

\begin{document}

\maketitle

\begin{abstract} By a probabilistic construction, we find a bipartite graph having 
average degree $d$ which can be expressed
as a conjunctive normal form using $C \log d$ clauses. This contradicts research problem 
1.33 of Jukna.

\end{abstract}

\section{Introduction}

We say $G=(V,W,E)$ is a bipartite graph over $V$ and $W$ if $V$ and $W$ are sets of 
vertices and $E \subset V \times W$ is the set of edges. Given two graphs $G_1$ and $G_2$ 
over $V$ and $W$ with $G_1=(V,W,E_1)$ and $G_2=(V,W,E_2)$,
we may define union and intersection edge-setwise, where
$$G_1 \cup G_2 = (V,W, E_1 \cup E_2),$$
and
$$G_1 \cap G_2=(V,W,E_1 \cap E_2).$$
We may define unions and intersections of families of bipartite graphs over $V$ and $W$.

A special type of graph we consider is $CL(A,B)$, the clause graph of $A \subset V$ and $B 
\subset W$. Then
$$CL(A,B)=\left (V,W,( A \times  W)  \cup ( V \times B) \right).$$
(The graph $CL(A,B)$ is called a clause graph because it is the union of all stars of vertices in
$A$ and $B$.)

We say that sets $A_1,\dots, A_n \subset V$ and  $B_1,\dots,B_n \subset W$ form a 
conjunctive normal form
using $n$ clauses
for a graph $G$ over $V$ and $W$ if
$$G = \bigcap_{i=1}^n  CL(A_i,B_i).$$

In Jukna's recent book \cite{Juk}, he poses the following conjecture as Research Problem 
1.33.

\bigskip

\begin{conjecture} \label{rubbish} There is a universal $\epsilon>0$ so that any bipartite graph $G$ 
having no $K_{2,2}$'s as subgraphs
and having average degree $d$ has no conjunctive normal form using $\lesssim d^{\epsilon}$ 
clauses. \end{conjecture}

A positive result for conjecture \ref{rubbish} would be important because it would allow one to construct a Boolean function
so that any low depth circuit computing it would have to have many gates. See (\cite{Juk}, Chapter 11).

Unfortunately, we prove

\bigskip

\begin{theorem} \label{main}  For all $\epsilon>0$ given $d$ sufficiently large, there is a bipartite graph $G$
with average degree  $\gtrsim d^{1-\epsilon}$ so that $G$ has 
a conjunctive normal form with at most $O(\log d)$ clauses. \end{theorem}

(Here we use the notation $A \gtrsim B$ to mean that there is a universal constant $C$, independent of $d$ so that
$CA \geq B$. We have stated theorem \ref{main} in this way because $d$ will be a parameter at the beginning
of our construction. Of course $\log d \sim  \log ( d^{1-\epsilon})$.)

Clearly, theorem \ref{main} contradicts conjecture \ref{rubbish}. Indeed, we remark that aside from constants, the
theorem is sharp. Given a $K_{2,2}$-free graph $G=(V,W,E)$ with average degree $d$, we may assume WLOG that there are
at least $d$ vertices $v_1,\dots v_d$ of $V$ adjacent to more than two elements of $W$ each. We let $W_v$ be
the set of elements of $W$ adjacent to $v$.  Then the sets $W_{v_1},\dots, W_{v_d}$ are distinct since in particular
each intersection of two of them contains at most one element by the $K_{2,2}$-free condition. However, if we have
$$G = \bigcap_{i=1}^n  CL(A_i,B_i),$$
then we have
$$W_{v} = \bigcap_{i: v \notin A_i}  B_{i}.$$
Thus there are at most $2^n$ distinct sets $W_v$. Hence $n \geq \log_2 d$.

We now explain the idea behind theorem \ref{main}. We consider the simplest model of a random bipartite graph
between sets of vertices having $N$ elements each. We choose i.i.d. Bernoulli random variables
$X_{v,w}$ indexed by $V \times W$. We define the random graph
$$G=(V,W,E),$$
where
$$E=\{(v,w):  X_{v,w}=1 \}.$$
To get average degree close to $d$, we set the probability that a given $X_{v,w}=1$ to be ${d \over N}$.
We should imagine that $N$ is quite large compared to $d$, say $N=d^{10}$. We calculate the probability that
there is a $K_{2,2}$ involving vertices $v_1,v_2,w_1,w_2$. By the independence of the random variables, clearly
the probability is ${d^4 \over N^4}$. Thus we expect the graph $G$ to have only $d^4$ copies of $K_{2,2}$. But this
is quite small compared to the number of vertices of $G$. By removing $2d^4$ vertices, we should be able to
get a $K_{2,2}$-free graph.

To prove theorem \ref{main}, we will replace this simple model of a random graph by a random conjunctive normal
form. We will show that it has roughly the same behavior as the random graph so that after removing a small number 
of vertices, which
we can do without changing the number of clauses in the conjunctive normal form, we arrive at a $K_{2,2}$-free
graph.

Finally, we make the remark that a simple argument using Cauchy-Schwarz shows that to get a $K_{2,2}$-free graph
of average degree $d$ on $N$ vertices, we need $N \gtrsim d^2.$ We remark that this Cauchy-Schwarz argument in fact imposes
a great deal of structure on the graph $G$. This lends us the temerity to make the following conjecture:

\bigskip

\begin{conjecture} \label{maybenotrubbish} There is a universal $\epsilon>0$ so that any bipartite graph $G$ 
having no $K_{2,2}$'s as subgraphs
and having average degree $d$ and fewer than $d^{2+\epsilon}$ vertices has no conjunctive normal form 
using $\lesssim d^{\epsilon}$ 
clauses. \end{conjecture}

{\bf Acknowledgements:}  The author is partially supported by NSF grant DMS-1001607
and a fellowship from the Guggenheim foundation. He would like to thank Esfandiar Haghverdi
for helpful discussions.

\section{Main Argument}

We now begin our proof of theorem \ref{main}. We start by defining a random conjunctive normal form, designed
to have average degree around $d$ with $V$ and $W$ being set of size $N=d^{10}$. We pick $p$ to be small but independent
of $d$. (Choosing $p={1 \over 100}$ would suffice.) Now we define i.i.d. Bernoulli random variables
$X_{j,v}$ and $Y_{j,w}$ indexed respectively by $\{1, \dots , n\} \times V$ and  
$\{1, \dots , n\} \times W$. We set the probability for each of $X_{j,v}$ and $Y_{j,w}$ to be 1 to be 
$p$.  Now we define 
$$A_i=\{v : X_{i,v}=0 \},$$
and
$$B_i=\{w: Y_{i,w}=0 \}.$$

We choose $n$ so that 
\begin{equation}  \label{degree} (1-p^2)^n \sim {d \over N}. \end{equation}
We achieve equation \ref{degree} by picking $n$ to be the nearest integer to $({1 \over p^2}) \ln ({N \over d}).$
In particular, this means that $n$ is $O(\log d)$.
We let
$$G = \bigcap_{i=1}^n  CL(A_i,B_i).$$
We will show that after a little pruning, we can modify $G$ to have no $K_{2,2}'$ and still have
average degree of at least $d^{1-\epsilon}.$

We now investigate the number of $K_{2,2}$'s in the graph $G$.

\bigskip

\begin{lemma} \label{countrectangles} Let $G$ be above. Let $v_1,v_2 \in V$ distinct and
$w_1,w_2 \in W$ distinct. The probability that there is a $K_{2,2}$ in $G$ on the vertices
$v_1,w_1,v_2,w_2$ is at most ${d^{4-\delta} \over N^{4-\delta}},$ where $\delta$ is small
depending only on $p$. \end{lemma}

\begin{proof} We observe that $v_1,w_1,v_2,w_2$ fail to be a $K_{2,2}$ only when there is some $j$ for which
one of $(v_1,w_1),(v_1,w_2),(v_2,w_1),(v_2,w_2)$ lies in the product $A_j^c \times B_j^c$. These are independent
events for different $j$. Now using inclusion-exclusion, we easily see that the probability that a $K_{2,2}$ is not
ruled out by the $j$th clause is $1-4p^2 + O(p^3)$. Now in light of equation \ref{degree}, the lemma is proved
\end{proof}

The reader should note that it is here that we have seriously used the presence of more than $\log d$ clauses.
The lemma doesn't work unless $p$ is small.

We still need to ensure that most vertices of the graph have a lot of degree.

\bigskip

\begin{lemma} \label{countedges} Let $G$ be as above. Let $\epsilon>0$ and $d$ sufficiently large. Let $v \in V$.
Then the probability that the degree $d_v$ of $v$ is satisfies 
$$ d^{1-\epsilon} \lesssim d_v \lesssim d^{1+\epsilon}$$ is at  
least ${9 \over 10}$. 
\end{lemma}

We delay the proof of lemma \ref{countedges} to point out why lemmas \ref{countrectangles} and \ref{countedges}
imply theorem \ref{main}. In light of lemma \ref{countedges}, the expected number of vertices of $V$ having degree 
 $\gtrsim d^{1-\epsilon}$ edges is at least ${9N \over 10}$. Therefore, with probability at least ${4 \over 5}$, the
graph $G$ has at least ${N \over 2}$ vertices in $V$ with degree  $\gtrsim d^{1-\epsilon}$. On the other hand
from lemma \ref{countrectangles}, the expected number of $K_{2,2}$'s is at most $N^{\delta} d^{4-\delta}$ which
by picking $p$ sufficiently small is bounded by $d^5$. Thus with probability ${1 \over 2}$ there are at most $2d^5$
copies of $K_{2,2}$ in $G$.  Thus there exists an instance of $G$ with ${N \over 2}$ vertices of $V$ having degree
$\gtrsim d^{1-\epsilon}$ and having at most $2 d^5$ copies of $K_{2,2}$. Let $V^{\prime}$ be the set of vertices
having degree t $\gtrsim d^{1-\epsilon}$ and not participating in any $K_{2,2}$'s.  Define
$$G^{\prime}=(V^{\prime}, W, E^{\prime}),$$
where 
$$E^{\prime} = \bigcap_{i=1}^n \left ( (A_i \cap V^{\prime}) \times W) \cup (V^{\prime} \times B_i)  \right).$$
Then $G^{\prime}$ satisfies the conclusion of theorem \ref{main}.

It remains to prove lemma \ref{countedges}. This will be a relatively simple application of the Chernoff-Hoeffding
bounds. We shall use the following simple form of them.

\bigskip

\begin{proposition} \label{Chernoff}  Given $M$ i.i.d. Bernoulli variables $X_1, \dots X_M$,
where the probability of $X_j=1$ being $p$, then if $q$ is the probability that
$$| (\sum_{j=1}^M X_j ) - pM  |  \geq \mu M ,$$
then
$$q \leq 2 e^{-2 \mu^2 n}.$$
\end{proposition}

Proposition \ref{Chernoff} follows from the results in \cite{Hoeff}.

Now we investigate the degree of a vertex $v$ in $G$. We let $W(v)$ be the set of vertices in $W$
which are adjacent to $v$. By the definition of $G$, we have that
$$W(v)=\bigcap_{i: v \notin A_i}  B_i.$$
In light of proposition \ref{Chernoff} there is a universal constant $C$ so that with probability ${19 \over 20}$
we have that
$$pn - C \sqrt{n}  \leq |\{i: v \notin A_i\}|  \leq pn + C \sqrt{n}.$$
We denote $m= |\{i: v \notin A_i\}|$ and denote by $i_1, \dots i_m$ the elements of $\{i: v \notin A_i\}$.
From now on, we work in the case
$$pn - C \sqrt{n}  \leq m \leq pn + C \sqrt{n}.$$
We name the sizes of the partial intersections
$$d_j = | \bigcap_{l=1}^j A_{i_l} |.$$
then $d_m$ is the degree of $v$. 
Now, in light of proposition \ref{Chernoff} we have for $d$ sufficiently large that with probability
at least $1-{1 \over 20 n}$, as long as $d_{j-1} \geq d^{1 \over 2}$, we have that
$$(1- p-d^{-{1 \over 6}}) d_{j-1}  \leq  d_j \leq (1- p+d^{-{1 \over 6}}) d_{j-1}.$$
Thus by induction, we see that as long as we are in the case where all these events hold, which has
probabiliy at least ${9 \over 10}$, we have the inequality
$$    N(1-p-d^{-{1 \over 6}})^{pn + C \sqrt{n} } \leq  d_m \leq N(1-p+d^{-{1 \over 6}})^{pn - C \sqrt{n} },$$
which for $d$ sufficiently large, we can rewrite as
$$ N d^{-\epsilon} (1-p)^{pn} \leq d_m \leq N d^{\epsilon} (1-p)^{pn},$$
which in light of equation \ref{degree} implies the desired result:
$$d^{1-\epsilon} \lesssim d_m \lesssim d^{1+\epsilon}.$$

\bigskip
\bigskip

\tiny

\textsc{N. KATZ, DEPARTMENT OF MATHEMATICS, INDIANA UNIVERSITY, BLOOMINGTON IN}

{\it nhkatz@indiana.edu}

\bigskip


\begin{thebibliography}{5}

\vskip.125in

\bibitem[Hoeff]{Hoeff} W. Hoeffding 
{\it Probability inequalities for sums of bounded random variables }
Journal of the American Statistical Association (1966)  {\bf 58}   13 - 30

\bibitem[Juk]{Juk} S. Jukna  {\it  Boolean Function Complexity: Advances and Frontiers}  
Springer,  Algorithms and Combinatorics
(2012)

\end{thebibliography}
\end{document}